\documentclass[journal,onecolumn,12pt, draftclsnofoot]{IEEEtran}
\usepackage{amsfonts,color,morefloats}
\usepackage{booktabs}
\usepackage{amssymb,amsmath,latexsym,amsthm}
\usepackage{graphicx}
\usepackage{makecell}
\usepackage{hyperref}
\usepackage{multirow}
\usepackage{array}

\newtheorem{theorem}{Theorem}[section]
\newtheorem{lemma}[theorem]{Lemma}

\newtheorem{problem}{Open Problem}

\newtheorem{example}{Example}

\newcommand{\C}{\mathcal{C}}

\begin{document}
	
	\title{Five infinite families of binary cyclic codes and their related codes with good parameters
	\thanks{The work of Hai Liu and Chengju Li was supported by the National
Natural Science Foundation of China (12071138), Shanghai Rising-Star Program (22QA1403200), the open research fund of National Mobile Communications Research Laboratory of Southeast University (2022D05), and the Shanghai Trusted Industry Internet Software Collaborative Innovation Center.
The work of Cunsheng Ding was supported by the Hong Kong Research Grants Council under Grant 16302121. \emph{(Corresponding author: Chengju Li.)}}}
	
\author{Hai Liu,~~ Chengju Li, ~~ Cunsheng Ding
\thanks{H. Liu and C. Li are with the Shanghai Key Laboratory of Trustworthy Computing, East China Normal University,
Shanghai, 200062, China; and are also with the National Mobile Communications Research Laboratory, Southeast University, Nanjing 210096, China
(email: 52265902014@stu.ecnu.edu.cn, cjli@sei.ecnu.edu.cn).}

\thanks{C. Ding is with the Department of Computer Science and Engineering, The Hong Kong University of Science and Technology,
Clear Water Bay, Kowloon, Hong Kong, China (email: cding@ust.hk).}}

	\date{\today}
	\maketitle
	
\begin{abstract}
Cyclic codes are an interesting type of linear codes
and have wide applications in communication and storage systems
due to their efficient encoding and decoding algorithms. Inspired by the recent work on binary cyclic codes published in IEEE Trans. Inf. Theory, vol. 68, no. 12,
pp. 7842-7849, 2022,
the objectives of this paper are the construction and analyses of five infinite families of
binary cyclic codes with parameters $[n, k]$ and $(n-6)/3 \leq k \leq 2(n+6)/3$. Three of the five families of binary cyclic codes
and their duals have a very good lower bound on their minimum distances and contain distance-optimal codes.
The other two families of binary cyclic codes are composed of binary duadic codes with a square-root-like lower bound on their minimum distances. As a by-product, two infinite families of self-dual binary codes with a square-root-like lower bound on their minimum distances are obtained.
	\end{abstract}
	

\newpage 	
	\tableofcontents

\newpage	
	\section{Introduction} \label{sec-intro}
	
	In this paper, let $\Bbb F_q $ denote the finite field of order $q$, where $q$ is a power of a prime $p$.
	An $[n,k,d]$ linear code $\mathcal C$ over $\Bbb F_q$ is a $k$-dimensional subspace of $\Bbb F_q^n$ with minimum (Hamming) distance $d$. The dual code of $\mathcal{C}$, denoted by $\mathcal{C}^{\perp}$, is defined by
$$\mathcal{C}^{\perp}=\{\mathbf{b} \in\Bbb F_q^n \ :\ \mathbf{b}\mathbf{c}^{T}=0 \,\,\text{for all $\mathbf{c} \in \mathcal{C}$}\}, $$
where $\mathbf{b}\mathbf{c}^{T}$ is the standard inner product of two vectors $\mathbf{b}$ and $\mathbf{c}$ in $\Bbb F_q^n$.

The linear code $\mathcal C$ over $\Bbb F_q$ is said to be \emph{cyclic} if $\mathbf (c_0, c_1, \ldots, c_{n-1}) \in \mathcal C $ implies $\mathbf (c_{n-1}, c_0, \ldots, c_{n-2}) \in \mathcal C $.
	By identifying each vector $\mathbf (c_0, c_1, \ldots, c_{n-1}) \in \Bbb F_q^n $ with
	$$c_0+c_1x+c_2x^2+\cdots+c_{n-1}x^{n-1} \in \Bbb F_q[x]/(x^n-1),$$
	a code $\mathcal C$ of length $n$ over $\Bbb F_q$ corresponds to a subset of $\Bbb F_q[x]/(x^n-1)$. Then $\mathcal C$ is a cyclic code if and only if the corresponding subset is an ideal of $\Bbb F_q[x]/(x^n-1)$.
	Note that every ideal of $\Bbb F_q[x]/(x^n-1)$ is principal. Then there is a monic polynomial $g(x)$ of the smallest degree such that $\mathcal C=( g(x) )$ and $g(x) \mid (x^n-1)$. Then $g(x)$ is called the \emph{generator
	polynomial} and $h(x)=(x^n-1)/g(x)$ is referred to as the \emph{check polynomial} of $\mathcal C$.
	Throughout this paper, assume that $\gcd(q,n)=1$.
 Denote $m=\text{ord}_n(q)$, i.e., $m$ is the smallest positive integer such that $q^m \equiv 1 \pmod n$.
Let $\alpha$ be a primitive element of $\Bbb F_{q^m}$ and put $\beta=\alpha^{\frac {q^m-1} n}$. Then
$\beta$ is a primitive $n$-th root of unity. The set $T=\{0 \le i \le n-1 : g(\beta^i)=0\}$ is referred to as the
\emph{defining set} of $\mathcal C$ with respect to $\beta$. If $T$ contains $\delta-1$ consecutive integers, then we have the
well-known BCH bound on cyclic codes, i.e., $d \geq \delta$.


Cyclic codes are interesting in theory, as they are closely related to quite a number of areas of mathematics such as
algebra, algebraic number theory, number theory, combinatorics and finite geometry. For example,
the determination of the weight distributions of irreducible cyclic codes is the same as the evaluation of certain Gaussian periods \cite{DY}.
Cyclic codes are important in practice due to their efficient encoding and decoding algorithms. However,
it is theoretically hard to design cyclic codes of length $n$ with good parameters if $n$ has small divisors
more than $1$ due to some general theory developed in \cite{Xiong,XiongZhang}. This fact is also
confirmed by the tables of best binary cyclic codes in Appendix A.2 of \cite{DingBK18}. It is harder to
design binary cyclic codes with good parameters as the alphabet size is too small. It is a much more difficult problem
to design an infinite family of binary cyclic codes such that each code in the family has good parameters.

It is an interesting problem to design an infinite family of binary cyclic codes with good parameters such that their duals also  have
good parameters. For convenience,  we call such an infinite family of binary cyclic codes a \emph{dually-good infinite family}
of binary cyclic codes.  Dully-good infinite families of binary cyclic codes with small or large dimensions relative to their lengths
are relatively easy to construct. The binary Hamming codes and the punctured binary second-order Reed-Muller codes are two dully-good infinite families of binary cyclic codes with a large and small dimension, respectively. However, only a small number of dully-good infinite families of binary cyclic codes with parameters $[n, k]$ and $(n-6)/3 \leq k \leq 2(n+6)/3$ are known in the literature. The following is a list of such binary cyclic codes known to the authors.
\begin{enumerate}
\item The family of binary quadratic-residue codes.
\item The punctured binary Reed-Muller codes of length $2^m-1$ and order $(m-1)/2$, where $m$ is odd.
\item Two families  of cyclic codes presented in \cite{TD}.
\end{enumerate}
It is in general very hard to determine the minimum distance of a cyclic code with parameters $[n, k]$ and $(n-6)/3 \leq k \leq 2(n+6)/3$ as the dimension $k$ is neither small nor large compared with a large length $n$ \cite{Charpin,LSX, JPen}. If it is impossible to
determine the minimum distance of such a code, the best one can do is to develop a good lower bound on the minimum distance of the code. This is the only way to show that such a code has a good error-correcting capability. However, this is not easy either. It is more difficult to develop good lower bounds on the minimum distances of both $\C$ and $\C^\perp$ \cite{GDL,SYW}. This explains why it is very difficult to find
a dully-good infinite family of binary cyclic codes with parameters $[n, k]$ and $(n-6)/3 \leq k \leq 2(n+6)/3$.

Inspired by the works in \cite{TD}, our objectives in this paper are the construction and analyses of five dully-good  infinite families of
binary cyclic codes with parameters $[n, k]$ and $(n-6)/3 \leq k \leq 2(n+6)/3$.
Three of the five infinite families of binary cyclic codes
and their duals have a very good lower bound on their minimum distances and contain distance-optimal codes.  The other two families of binary cyclic codes are composed of binary duadic codes with a square-root-like lower bound on their minimum distances. As a by-product, two families of self-dual binary codes with a square-root-like lower bound on their minimum distances are obtained.

 In this paper by the Database we mean the tables of best known linear codes  \cite{G}, which are maintained by Markus Grassl at
http://www.codetables.de/. We inform the reader that all the code examples given in this paper are computed by the Magma software package.

The rest of this paper is organised as follows. Section \ref{sec-3fcyclic} introduces and analyses the first three families of
binary cyclic codes and their related codes.  Section \ref{sec-duadicc} constructs and analyses the two families of
binary duadic codes and their related codes.  Section \ref{sec-Con} concludes this paper and proposes some open problems.

\section{The first three families of binary cyclic codes and their duals}\label{sec-3fcyclic}

\subsection{The construction of the first three families of binary cyclic codes }

Let $m \geq 2$ be a positive integer and let $n = {2^m -1}$.
Let  $\Bbb Z_n =\{0,1,2,\ldots,n-1\}$ be the ring of integers modulo $n$. For any $s \in \Bbb Z_n$, the $2$-cyclotomic coset of $s$ modulo $n$ is defined by
	$$C_s^{(2,n)}=\{s,s2,s2^2,\ldots,s2^{l_s-1}\}\bmod n \subseteq \Bbb Z_n,$$
	where $l_s$ is the smallest positive integer such that $s \equiv s2^{l_s}\pmod  n$.
For an integer $i$ with $0 \le i \le 2^m-1$, let $$i=i_{m-1}2^{m-1} + i_{m-2}2^{m-2} + \cdots + i_1 2+i_0$$ be the $2$-adic expansion of $i$, where
	$i_j \in \{0, 1\}$ for $0 \le j \le m-1$. We will also write $i=(i_{m-1},i_{m-2},\ldots,i_1,i_0)$ and call it the $2$-adic expansion of $i$ in the sequel.
For any $i$ with $ 0 \leq i \leq n-1 $, define $w_{2}(i) = \sum\limits_{j=0}^{m-1}i_{j}$. 	
	Let $\alpha$ be a generator of $\Bbb F_{2^m}^*$. We define a polynomial
	\begin{equation} \label{equ-gim}
		g_{(i,m)}(x) =  \prod\limits_{
			\substack{1 \leq i \leq n-1\\
				w_{2}(j) \equiv i\pmod{3}}}
		(x-\alpha^{j})
	\end{equation}
	for each $i \in \left\{{0,1,2}\right\}$.
Note that $w_{2}(a_1)=w_{2}(a_2)$ if $a_1$ and $a_2$ belong to the same cyclotomic coset.
 It then follows that $g_{(i,m)}(x) \in \Bbb F_2[x]$. Let $\mathcal C_{(i,m)}$ denote the binary cyclic code of length $n=2^m-1$ with generator polynomial $g_{(i,m)}(x)$ for $i=0, 1, 2$.
When $m=3$ and $4$, Table \ref{Tab} shows that the
three families of binary cyclic codes  and their dual codes contain some optimal binary
cyclic codes, which motivated us to study the parameters of the three families of binary cyclic codes $\mathcal C_{(i,m)}$ and their dual codes.

	 \begin{table}
     \centering
     \caption{Parameters of $\mathcal{C}_{(i,m)}$ for $m=3, 4$} \label{Tab}
     \begin{tabular}{|c|c|c|c|c|c|c|c|c|c|c|}
      \hline
      Code  & Parameters  & Optimality  &&  Code  & Parameters  & Optimality \\ \hline
      $\mathcal C_{(0,3)}$ & $ [7,7,1] $ &     Optimal &  & $\mathcal C_{(0,4)}$ & $ [15,11,3] $ &     Optimal   \\ \hline
      $\mathcal C_{(1,3)}$ & $ [7,4,3] $ &    Optimal &  & $\mathcal C_{(1,4)}$ & $ [15,11,3] $ &    Optimal  \\ \hline
      $\mathcal C_{(2,3)}$ & $ [7,4,3] $ &    Optimal &  & $\mathcal C_{(2,4)}$ & $ [15,9,4] $ &    Optimal   \\ \hline
      $\mathcal C_{(0,3)}^\bot$ & $[7,0]$ &  / &  & $\mathcal C_{(0,4)}^\bot$ & $[15,4,8]$ &  Optimal    \\ \hline
      $\mathcal C_{(1,3)}^\bot$ & $[7,3,4]$ &  Optimal &  & $\mathcal C_{(1,4)}^\bot$ & $[15,4,8]$ &  Optimal   \\ \hline
      $\mathcal C_{(2,3)}^\bot$ & $[7,3,4]$ &  Optimal &  & $\mathcal C_{(2,4)}^\bot$ & $[15,6,6]$ &  Optimal   \\ \hline
     \end{tabular}
    \end{table}

  In the following three subsections, the dimensions of the three families of binary cyclic codes $\mathcal C_{(i,m)}$ and their dual codes will be determined and
lower bounds on their minimum distances will be developed with the BCH bound on cyclic codes.

\subsection{Some auxiliary results}\label{sec-auxi}
	
In this subsection, we will present some auxiliary results on the defining sets of the three families of binary cyclic codes $\mathcal C_{(i,m)}$, which will play an important role
in developing good lower bounds on the minimum distance of the binary codes.
The following well-known lemma will be employed later.

	\begin{lemma}\label{lemma-gcd2}
		Let $l$ and $m$ be two positive integers. Then
		$$\gcd(a^{m}-1, a^{l}-1) = a^{\gcd(m,l)} -1,$$
		where $a \geq 2$ is a positive integer.
	\end{lemma}

 Below we always assume that $n=2^m-1$. Denote
$$T_{(i,m)} = \{ 1 \leq j \leq n-1:w_{2}(j) \equiv i \pmod{3}\}$$
for $i=0, 1, 2$. By the definition of $\mathcal C_{(i,m)}$ given by \eqref{equ-gim}, $T_{(i,m)}$ is the defining set of $\mathcal C_{(i,m)}$ with respect to the $n$-th primitive root of unity $\alpha$.

\subsubsection{The odd $m$ case}

	\begin{lemma}\label{lemma-m61}
		Let $m \equiv 1 \pmod{6} \geq 7$. Then we have the following.
\begin{enumerate}
  \item If $v=2^{(m-1)/2} -1$, then $\gcd(v,n)=1$ and
		$$\{av : 1 \leq a \leq 2^{(m-1)/2} +2 \} \subseteq	T_{(0,m)}.$$
  \item If $v=2^{(m+1)/2} -1$, then $\gcd(v,n)=1$ and $$\{av : 1 \leq a \leq 2^{(m-1)/2} +2 \} \subseteq	T_{(1,m)}.$$
  \item If $v=2^{(m-1)/2} -1$, then $\gcd(v,n)=1$ and $$\{av  \bmod \ n: 2^{m-1}+2^{(m-1)/2} +1 \leq a \leq  2^{m-1}+2^{(m-1)/2} + 2^{(m-3)/2}+ 1 \} \subseteq	T_{(2,m)}.$$
\end{enumerate}		
	\end{lemma}
	
	\begin{proof}
If $v=2^{(m-1)/2} -1$, it follows from Lemma \ref{lemma-gcd2} that $ \gcd(v,n) =1$.
		 When $a=2^{(m-1)/2} +2$, we have
		$$av =2^{m-1} + 2^{(m-1)/2} -2 = 2(2^{m-2} + 2^{(m-3)/2}-1).$$
		Consequently, $w_{2}(av) =(m-1)/2 \equiv 0 \pmod{3}$.
		When $ a=2^{(m-1)/2} +1 $, $ av = 2^{m-1} -1$ and $w_{2}(av) =(m-1)\equiv 0 \pmod{3}$.
		When $ a=2^{(m-1)/2}$, $w_{2}(av) =w_{2}(v)=(m-1)/2\equiv 0 \pmod{3}$.
		Now we assume that $1 \leq a \leq 2^{(m-1)/2} - 1$. Let $a=2^{l}\bar{a}$, where $\bar{a}$ is odd and $l \ge 0$ is an integer. Then we have $1 \leq a \leq 2^{(m-1)/2} - 1$ and the $2$-adic expansion of $\bar{a}$ given by $$ \bar{a} = \sum\limits_{i=0}^{(m-3)/2}a_{i}2^{i}.$$
		Since $\bar{a}$ is odd, $a_{0}=1$. We have
		\begin{align*}
			\bar{a}v =\sum\limits_{i=1}^{(m-3)/2}a_{i}2^{i+(m-1)/2} + \sum\limits_{i=0}^{(m-3)/2}(1-a_{i})2^{i} +1.
		\end{align*}
	    It then follows that
		$$w_{2}(\bar{a}v) =  w_{2}(\bar{a})-1 + 1+ \frac{m-1}{2}-w_{2}(\bar{a}) = \frac{m-1}{2}\equiv 0 \pmod{3}.$$
		The desired conclusion in the first case then follows.

		If $v=2^{(m+1)/2} -1$, it follows from Lemma \ref{lemma-gcd2} that
		$$\gcd(v,n)=2^{\gcd((m+1)/2,m)}-1 =2^{\gcd((m+1)/2,(m-1)/2)}-1 =1.$$
	When $ a=2^{(m-1)/2}$, it is easy to see that
		$$w_{2}(av) = w_{2}(v) = \frac{m+1}{2}\equiv 1 \pmod{3}.$$
Furthermore, one can similarly check that $w_{2}(av) \equiv 1 \pmod{3}$ for $a=2^{(m-1)/2}+1$ and $2^{(m-1)/2}+2$.
		Next, we assume that $1 \leq a \leq 2^{(m-1)/2} - 1$. Let $a=2^{l}\bar{a}$, where $\bar{a}$ is odd and $l \ge 0$ is an integer. Then we have $1 \leq \bar{a} \leq 2^{(m-1)/2} - 1$. Let the 2-adic expansion of $\bar{a}$ be given by
		$$\bar{a} = \sum\limits_{i=0}^{(m-3)/2}a_{i}2^{i}.$$
		Since $\bar{a}$ is odd, $a_{0}=1$.  Then
		\begin{align*}
			\bar{a}v &= \bar{a}2^{(m+1)/2} -\bar{a} \\
			&= \sum\limits_{i=1}^{(m-3)/2}a_{i}2^{i+(m+1)/2} +2^{(m-1)/2} + \sum\limits_{i=0}^{(m-3)/2}(1-a_{i})2^{i} +1.
		\end{align*}
		As a result, we have
		$$ w_{2}(\bar{a}v) =  w_{2}(\bar{a})-1 + 2 + \frac{m-1}{2}-w_{2}(\bar{a}) = \frac{m+1}{2}\equiv 1 \pmod{3}.$$
	
If $v=2^{(m-1)/2} -1$, it follows from Lemma \ref{lemma-gcd2} that
		$\gcd(v,n) =1$. When $a=2^{m-1}+2^{(m-1)/2} +1$,
one can check that
		$$w_{2}(av) =w_{2}(v)=(m-3)/2\equiv 2\pmod{3}.$$
		Write $a= 2^{m-1}+2^{(m-1)/2} +1 + t$ for $1 \leq t \leq 2^{(m-3)/2} - 1$ and $t = \sum\limits_{i=0}^{(m-5)/2}t_{i}2^{i}$.
		Then we have
		\begin{align*}
			av &= (2^{m-1}+2^{(m-1)/2} +1 + \sum\limits_{i=0}^{(m-5)/2}t_{i}2^{i})(2^{(m-1)/2} -1) \\
			   &= 2^{(m-1)/2}2^{m-1} -1 + \sum\limits_{i=0}^{(m-5)/2}t_{i}2^{i+(m-1)/2} -\sum\limits_{i=0}^{(m-5)/2}t_{i}2^{i} \\
               &\equiv \sum\limits_{i=0}^{(m-5)/2}t_{i}2^{i+(m-1)/2} +  \sum\limits_{i=0}^{(m-5)/2}(1-t_{i})2^{i} \pmod n.
		\end{align*}
		    It then follows that
		$$ w_{2}(av \bmod{n})= w_{2}(t) + \frac{m-3}{2} - w_{2}(t) = \frac{m-3}{2} \equiv 2 \pmod{3}. $$
		This completes the proof.
	\end{proof}

	\begin{lemma}\label{lemma-m63}
		Let $m \equiv 3 \pmod{6} \geq 9$. Then we have the following.
		\begin{enumerate}
			\item If $v=2^{(m-1)/2} -1$, then $\gcd(v,n)=1$ and
			$$\{av \bmod{n}: 2^{m-1}+2^{(m-1)/2} +1 \leq a \leq  2^{m-1}+2^{(m-1)/2} + 2^{(m-3)/2} \} \subseteq T_{(0,m)}.$$
			\item If $v=2^{(m-1)/2} -1$, then $\gcd(v,n)=1$ and $$\{av : 1 \leq a \leq 2^{(m-1)/2} \} \subseteq	T_{(1,m)}.$$
			\item If $v=2^{(m+1)/2} -1$, then $\gcd(v,n)=1$ and $$\{av  : 1 \leq a \leq 2^{(m-1)/2} \} \subseteq	T_{(2,m)}.$$
		\end{enumerate}		
	\end{lemma}
    \begin{proof}
    	The proof is very similar to that of Lemma \ref{lemma-m61} and omitted here.
    \end{proof}

	\begin{lemma}\label{lemma-m65}
		Let $m \equiv 5 \pmod{6} \geq 5$. Then we have the following.
		\begin{enumerate}
			\item If $v=2^{(m+1)/2} -1$, then $\gcd(v,n)=1$ and
			$$\{av : 1 \leq a \leq  2^{(m-1)/2}  \} \subseteq	T_{(0,m)}.$$
			\item If $v=2^{(m+1)/2} -1$, then $\gcd(v,n)=1$ and $$\{av \bmod{n}: 2^{m-1}+2^{(m-1)/2} +1 \leq a \leq  2^{m-1}+2^{(m-1)/2} + 2^{(m-3)/2} \} \subseteq	T_{(1,m)}.$$
			\item If $v=2^{(m-1)/2} -1$, then $\gcd(v,n)=1$ and $$\{av : 1 \leq a \leq 2^{(m-1)/2}  \} \subseteq	T_{(2,m)}.$$
		\end{enumerate}		
	\end{lemma}
	\begin{proof}
		The proof is very similar to that of Lemma \ref{lemma-m61} and omitted here.
	\end{proof}

\subsubsection{The even $m$ case}
	
	\begin{lemma}\label{lemma-m401}
		Let $m\equiv 0 \pmod{4} \geq 4$  and $v=2^{(m-2)/2} - 1$. Then $\gcd(v,n)=1$ and
		$$
			\{av : 1 \leq a \leq 2^{(m-2)/2}\} \subseteq\begin{cases}
				  T_{(2,m)}, & \text {if } m\equiv 0 \pmod{6}; \\
			      T_{(0,m)}, & \text {if } m\equiv 2 \pmod{6};  \\
				  T_{(1,m)}, & \text {if } m\equiv 4 \pmod{6}.
			\end{cases}
		$$
	\end{lemma}
	\begin{proof}
		It follows from Lemma \ref{lemma-gcd2} that
		$$ \gcd(v,n)=2^{\gcd((m-2)/2,m)}-1 =2^{\gcd((m-2)/2,(m+2)/2)}-1 =2^{\gcd((m-2)/2,2)}-1= 1.$$
		We begin to prove the second conclusion. When $a =2^{(m-2)/2}$, we have
		$$w_{2}(av) =w_{2}(v)=\frac{m-2}{2}\equiv \begin{cases}
			2 \pmod{3}, & \text{if } m \equiv 0 \pmod{6}; \\
			0 \pmod{3}, & \text{if } m\equiv 2 \pmod{6}; \\
			1 \pmod{3}, & \text{if }  m\equiv 4 \pmod{6}.
		\end{cases}$$
		Now assume that $1 \leq a \leq 2^{(m-2)/2} - 1$. Let $a=2^{l}\bar{a}$, where $\bar{a}$ is odd and $l \ge 0$ is an integer. Let the 2-adic expansion of $\bar{a}$ be given by
		$$ \bar{a} = \sum\limits_{i=0}^{(m-4)/2}a_{i}2^{i},$$
		where $a_{0}=1$. We have
		\begin{align*}
			\bar{a}v =\sum\limits_{i=1}^{(m-4)/2}a_{i}2^{i+(m-2)/2} + \sum\limits_{i=0}^{(m-4)/2}(1-a_{i})2^{i} +1.
		\end{align*}
		It then follows that
		\begin{align*} w_{2}(av)&=w_{2}(\bar{a}v) =  w_{2}(\bar{a})-1 + 1+ \frac{m-2}{2}-w_{2}(\bar{a})\\ &= \frac{m-2}{2}\equiv\begin{cases}
			2 \pmod{3}, & \text{if } m\equiv 0 \pmod{6}; \\
			0 \pmod{3}, & \text{if } m\equiv 2 \pmod{6};  \\
			1 \pmod{3}, & \text{if } m\equiv 4 \pmod{6}. \end{cases}\end{align*}
		This completes the proof.		
	\end{proof}

	\begin{lemma}\label{lemma-m421}
		Let $m\equiv 2 \pmod{4} \geq 4$  and $v=2^{(m-4)/2} - 1$. Then $\gcd(v,n)=1$ and
		\begin{align*}
			\{av : 1 \leq a \leq 2^{(m-4)/2} \} \subseteq \begin{cases}
				 T_{(1,m)}, & \text{if } m\equiv 0 \pmod{6}; \\
				 T_{(2,m)}, & \text{if } m\equiv 2 \pmod{6}; \\
			     T_{(0,m)}, & \text{if } m\equiv 4 \pmod{6}.
			\end{cases}
		\end{align*}
	\end{lemma}

	\begin{proof}
		It follows from Lemma \ref{lemma-gcd2} that
		$$ \gcd(v,n)=2^{\gcd((m-4)/2,m)}-1 =2^{\gcd((m-4)/2,(m+4)/2)}-1= 1.$$
		When $a =2^{(m-4)/2}$, we have
		$$w_{2}(av) =w_{2}(v)=\frac{m-4}{2}\equiv \begin{cases}
			1 \pmod{3}, & \text{if } m\equiv 0 \pmod{6};  \\
			2 \pmod{3}, & \text{if } m\equiv 2 \pmod{6};  \\
			0 \pmod{3}, & \text{if } m\equiv 4 \pmod{6}.
		\end{cases}$$
		Next, we assume that $1 \leq a \leq 2^{(m-4)/2} - 1$. Let $a=2^{l}\bar{a}$, where $\bar{a}$ is odd and $l \ge 0$ is an integer. Let the 2-adic expansion of $\bar{a}$ be given by
		$$ \bar{a} = \sum\limits_{i=0}^{(m-6)/2}a_{i}2^{i},$$
		where $a_{0}=1$. One can check that
		\begin{align*}
			\bar{a}v =\sum\limits_{i=1}^{(m-6)/2}a_{i}2^{i+(m-2)/2} + \sum\limits_{i=0}^{(m-6)/2}(1-a_{i})2^{i} +1.
		\end{align*}
		It then follows that
		$$ w_{2}(av)=w_{2}(\bar{a}v) = \frac{m-4}{2}\equiv\begin{cases}
			1 \pmod{3}, & \text{if } m\equiv 0 \pmod{6}; \\
			2 \pmod{3}, & \text{if } m\equiv 2 \pmod{6}; \\
			0 \pmod{3}, & \text{if } m\equiv 4 \pmod{6}. \end{cases}$$
		Summarizing the discussions above completes proof.		
	\end{proof}

\subsection{Parameters of the binary codes $\mathcal C_{(i,m)}$}\label{sec-sqpara}

In this subsection, we investigate the dimensions and   minimum distances of the binary cyclic codes $\mathcal C_{(i,m)}$ for $i=0, 1, 2$.

Let $\omega=e^{\frac {2 \pi \sqrt{-1}} 3} \in \Bbb C$ be a $3$-th primitive root of unity, i.e., $\omega^2+\omega+1=0$. Denote
\begin{equation} \label{equ-s012}s_0=\sum_{\substack{0 \leq i \leq m\\ i \equiv 0\pmod{3}}}\binom{m}{i}, \ s_1=\sum_{\substack{0 \leq i \leq m\\ i \equiv 1\pmod{3}}}\binom{m}{i}, \ \text{ and }
s_2=\sum_{\substack{0 \leq i \leq m\\ i \equiv 2 \pmod{3}}}\binom{m}{i}.\end{equation}
Then we have
$$(1+\omega)^m=s_0+s_1 \omega+s_2 \omega^2, \ \ (1+\omega^2)^m=s_0+s_1 \omega^2+s_2 \omega.$$
It follows from $(1+\omega)^m(1+\omega^2)^m=1$ that
\begin{equation} \label{equ-srelation}
s_0^2+s_1^2+s_2^2-s_0s_1-s_0s_2-s_1s_2=1.
\end{equation}

\subsubsection{Parameters of  $\mathcal C_{(i,m)}$ when $m \ge 5$ is odd}

	When $m \equiv 1 \pmod {6} \geq 7$ is odd, the parameters of  $\mathcal C_{(i,m)}$ for $i=0, 1, 2$ are investigated in the following theorem.

	\begin{theorem}\label{them-4.1}
		Let $m \equiv 1 \pmod{6} \geq 7$ be an odd integer. Then the codes $\mathcal C_{(0,m)}$ and $\mathcal C_{(1,m)}$ have parameters
$$[2^m-1, \ (2^{m+1}-1)/3, \ d \geq 2^{(m-1)/2}+ 3],$$  and
		the code $\mathcal C_{(2,m)}$ has parameters $$[2^m-1, \ (2^{m+1}-1)/3, \ d \geq 2^{(m-3)/2} +2].$$
	\end{theorem}

	\begin{proof}
		Note that $m \equiv 1 \pmod{6}$ is odd and  $w_{2}(i)+ w_{2}(n-i) = m$ for each $i$  with $ 1 \leq i \leq  n-1$. Hence $i \in  T_{(0,m)}$ if and only if $n-i \in T_{(1,m)}$, which means that
		$$|T_{(0,m)}|= |T_{(1,m)}|.$$
	In this case, by \eqref{equ-s012}, we have
$$s_0=|T_{(0,m)}|+1, \ s_1=|T_{(1,m)}|+1, \ s_2=|T_{(2,m)}|.$$
One can compute from \eqref{equ-srelation} that
$$(|T_{(0,m)}|+1-|T_{(2,m)}|)^2=1.$$
Then $$|T_{(0,m)}|= |T_{(2,m)}| \ \text{ or } \ |T_{(2,m)}|= |T_{(0,m)}|+2.$$
We assert that $|T_{(0,m)}|= |T_{(2,m)}|$ holds, otherwise, we have $|T_{(2,m)}|= |T_{(0,m)}|+2$ and $3|T_{(0,m)}|+2=2^m-2$, which leads to a contradiction since $|T_{(0,m)}|=(2^m-4)/3$ is not a integer.
 Consequently, $$|T_{(0,m)}| =|T_{(1,m)}| = |T_{(2,m)}| = (2^{m}-2)/3$$ and
$$\dim(\mathcal C_{(0,m)})=\dim(\mathcal C_{(1,m)})=\dim(\mathcal C_{(2,m)})=n-(2^{m}-2)/3=(2^{m+1}-1)/3.$$

  Denote $v=2^{\frac {m-1} 2}-1$. It follows from Lemma \ref{lemma-m61} that $\gcd(v,n) =1$. Let $ \overline{v}$ be the integer satisfying $v \overline{v} \equiv 1 \pmod n$.
  Write $\gamma = \alpha^{\overline{v}}$. It is deduced from Lemma \ref{lemma-m61} that the defining set of $\mathcal C_{(0,m)}$ with respect to $\gamma$ contains the set $\left\{ 1,2,...,2^{(m-1)/2}+2 \right\}$. The lower bound on the minimum distance of $\mathcal C_{(0,m)}$ then follows from the BCH bound on cyclic codes. The desired conclusion on the minimum distances of $\mathcal C_{(1,m)}$ and $\mathcal C_{(2,m)}$ can be similarly obtained. This completes the proof.
	\end{proof}

When $m \equiv 3 \pmod {6} \geq 9$ is odd, the parameters of  $\mathcal C_{(i,m)}$ for $i=0, 1, 2$ are treated in the following theorem.

	\begin{theorem}\label{them-4.2}
		Let $m \equiv 3 \pmod{6} \geq 9$ be an odd integer. Then the codes $\mathcal C_{(1,m)}$ and  $\mathcal C_{(2,m)}$ have  parameters
$$[2^m-1, \ (2^{m+1}-4)/3, \ d \geq 2^{(m-1)/2} +1],$$ and
		the code $\mathcal C_{(0,m)}$ has parameters $$[2^m-1, \ (2^{m+1}+5)/3, \ d \geq 2^{(m-3)/2} +1].$$
	\end{theorem}

	\begin{proof}
		Note that $m \equiv 3 \pmod{6}$ and  $w_{2}(i) + w_{2}(n-i)= m$ for each $i$ with $ 1 \leq i \leq  n-1$. Hence $i \in  T_{(1,m)}$ if and only if $n-i \in  T_{(2,m)}$.
Then $$|T_{(1,m)}|= |T_{(2,m)}|.$$ In this case, by \eqref{equ-s012}, we have
$$s_0=|T_{(0,m)}|+2, \ s_1=|T_{(1,m)}|, \ s_2=|T_{(2,m)}|.$$
One can similarly obtain from \eqref{equ-srelation} that
$|T_{(1,m)}|=|T_{(0,m)}|+3$.
Thus, we have $$|T_{(0,m)}| =(2^{m}-8)/3, \
|T_{(1,m)}| = |T_{(2,m)}| = (2^{m}+1)/3.$$
The desired conclusion on the dimensions then follows.
The lower bounds on the minimum distances can be derived by employing Lemma \ref{lemma-m63}. The proof is very similar to that of Theorem \ref{them-4.1} and omitted.
	\end{proof}

When $m \equiv 5 \pmod {6} \geq 5$ is odd, the parameters of  $\mathcal C_{(i,m)}$ for $i=0, 1, 2$ are investigated in the following theorem.

	\begin{theorem}\label{them-4.3}
		Let $m \equiv 5 \pmod{6} \geq 5$ be an odd integer. Then the codes $\mathcal C_{(0,m)}$ and $\mathcal C_{(2,m)}$ have  parameters
$$[2^m-1, \ (2^{m+1}-1)/3, \ d \ge 2^{(m-1)/2} +1],$$ and
		the code $\mathcal C_{(1,m)}$ has parameters $$[2^m-1, \ (2^{m+1}-1)/3, \ d \ge 2^{(m-3)/2} +1].$$
	\end{theorem}

	\begin{proof}
		Note that $m \equiv 5 \pmod{6}$. Then one similarly has
$$|T_{(2,m)}| = |T_{(0,m)}|$$ and
$$s_0=|T_{(0,m)}|+1, \ s_1=|T_{(1,m)}|, \ s_2=|T_{(2,m)}|+1.$$
The remainder of the proof is very similar to that of Theorem \ref{them-4.1} with the help of Lemma \ref{lemma-m65}, and we omit the details.
	\end{proof}

\subsubsection{Parameters of  $\mathcal C_{(i,m)}$ when $m \ge 4$ is even}

Assume that $m \geq 4$ is an even integer. We investigate the parameters of  $\mathcal C_{(i,m)}$ in the two cases:
$m \equiv 0 \pmod{4}$ and $m \equiv 2 \pmod{4}$.
When $m \equiv 0 \pmod{4} \geq 4$ is even, the parameters of  $\mathcal C_{(i,m)}$ for $i=0, 1, 2$ are studied in the following theorem.

    \begin{theorem}\label{them-4.5}
    Let $  m \equiv 0 \pmod{4} \geq 4$ be an even integer. Then we have the following.
    	\begin{enumerate}
    		\item If $m \equiv 0 \pmod{6} \geq 6$, then $\mathcal C_{(1,m)}$ and $\mathcal C_{(2,m)}$ have parameters $$[2^m-1, \ (2^{m+1}-2)/3, \ d   \geq 2^{(m-2)/2} +1],$$
    and $\mathcal C_{(0,m)}$ has parameters $$[2^m-1, \ (2^{m+1}+1)/3, \ d \geq 2^{(m-2)/2} +1].$$
    		\item If $m \equiv 2 \pmod{6} \geq 8$, then $\mathcal C_{(0,m)}$ and $\mathcal C_{(2,m)}$ have parameters $$[2^m-1, \ (2^{m+1}+1)/3, \ d  \geq 2^{(m-2)/2} +1],$$
    and $\mathcal C_{(1,m)}$ has parameters $$[2^m-1, \  (2^{m+1}-5)/3, \ d \geq 2^{(m-2)/2} +1].$$
    		\item If $m \equiv 4 \pmod{6} \geq 4$, then $\mathcal C_{(0,m)}$ and $\mathcal C_{(1,m)}$ have parameters $$[2^m-1, \  (2^{m+1}+1)/3, \  d \geq 2^{(m-2)/2} +1],$$
    and $\mathcal C_{(2,m)}$ has parameters $$[2^m-1, \ (2^{m+1}-5)/3, \ d \geq 2^{(m-2)/2} +1].$$
    	\end{enumerate}	
    \end{theorem}

    \begin{proof}
    	When $m \equiv 0 \pmod{6}$, we similarly have
$|T_{(1,m)}| = |T_{(2,m)}|$ and
$$s_0=|T_{(0,m)}|+2, \ s_1=|T_{(1,m)}|, \ s_2=|T_{(2,m)}|.$$
It is deduced from \eqref{equ-srelation} that   $$|T_{(0,m)}| =(2^{m}-4)/3, \
|T_{(1,m)}| = |T_{(2,m)}| = (2^{m}-1)/3.$$

    	When $m \equiv 2 \pmod{6} \geq 8$, we similarly have
    	$|T_{(0,m)}| = |T_{(2,m)}|$ and $$s_0=|T_{(0,m)}|+1, \ s_1=|T_{(1,m)}|, \ s_2=|T_{(2,m)}|+1.$$
It is deduced from \eqref{equ-srelation} that  $$|T_{(0,m)}|= |T_{(2,m)}| =(2^{m}-4)/3, \
|T_{(1,m)}|  = (2^{m}+2)/3.$$

    	When $m \equiv 4 \pmod{6} \geq 4$,  we similarly have
    	$|T_{(0,m)}| = |T_{(1,m)}|$  and $$s_0=|T_{(0,m)}|+1, \ s_1=|T_{(1,m)}|+1, \ s_2=|T_{(2,m)}|.$$
It is deduced from \eqref{equ-srelation} that  $$|T_{(0,m)}|= |T_{(1,m)}| =(2^{m}-4)/3, \
|T_{(2,m)}|  = (2^{m}+2)/3.$$
     The desired conclusion on the dimension of $\mathcal C_{(i,m)}$ then follows.

    For $v=2^{(m-2)/2} - 1$, it follows from Lemma \ref{lemma-m401} that $\gcd(v,n) =1$ if $m \equiv 0 \pmod{4}$.
     Let $ \overline{v}$ be the multiplicative inverse of $v$ modulo $n$  and let $\gamma = \alpha^{\overline{v}}$. It then follows again from Lemma \ref{lemma-m401} that the defining set of $\mathcal C_{(i,m)}$ with respect to $\gamma$ contains the set $\left\{ 1,2,...,2^{(m-2)/2} \right\}$. The desired lower bound on $d$ then follows from the BCH bound on cyclic codes.
    \end{proof}

When $m \equiv 2 \pmod{4} \geq 4$ is even, the parameters of  $\mathcal C_{(i,m)}$ for $i=0, 1, 2$ are investigated in the following theorem.

    \begin{theorem}\label{them-4.6}
     Let $  m \equiv 2 \pmod{4} \geq 4$ be an even integer. Then we have the following.
     \begin{enumerate}
    	\item If $m \equiv 0 \pmod{6} \geq 6$, then $\mathcal C_{(1,m)}$ and $\mathcal C_{(2,m)}$ have parameters $$[2^m-1, \ (2^{m+1}-2)/3, \ d \geq 2^{(m-4)/2} +1],$$
    and $\mathcal C_{(0,m)}$ has parameters $$[2^m-1, \  (2^{m+1}+1)/3, \ d \geq 2^{(m-4)/2} +1].$$
    	\item If $m \equiv 2 \pmod{6} \geq 8$, then $\mathcal C_{(0,m)}$ and $\mathcal C_{(2,m)}$ have  parameters $$[2^m-1, \  (2^{m+1}+1)/3, \  d \geq 2^{(m-4)/2} +1],$$
    and $\mathcal C_{(1,m)}$ has parameters $$[2^m-1, \   (2^{m+1}-5)/3, \ d \geq 2^{(m-4)/2} +1].$$
    	\item If $m \equiv 4 \pmod{6} \geq 4$, then $\mathcal C_{(0,m)}$ and $\mathcal C_{(1,m)}$ have  parameters $$[2^m-1, \  (2^{m+1}+1)/3, \  d \geq 2^{(m-4)/2} +1],$$
   and $\mathcal C_{(2,m)}$ has parameters $$[2^m-1, \ (2^{m+1}-5)/3, \ d \geq 2^{(m-4)/2} +1].$$
     \end{enumerate}
 	    \end{theorem}

    \begin{proof}
    	The proof is very similar to that of Theorem \ref{them-4.5} and omitted here.
    \end{proof}

\subsection{Parameters of the dual codes $\mathcal C_{(i,m)}^\bot$}\label{sec-sqparadual}

In this subsection, we investigate the parameters of  the dual codes $\mathcal C_{(i,m)}^\bot$ for $i=0, 1, 2$. Their dimensions are explicitly determined and low bounds on the minimum distances of these codes are developed.

	\subsubsection{Parameters of  $\mathcal C_{(i,m)}^\bot$ when $m \ge 5$ is odd}

To develop lower bounds on the minimum distances of $\mathcal C_{(i,m)}^\bot$ for odd $m \ge 5$, we need the following three lemmas,
which can be similarly proved by using the same techniques given in Section \ref{sec-auxi}.
Below we only state the lemmas and omit their proofs.

\begin{lemma}\label{lemma-m61d}
		Let $m \equiv 1 \pmod{6} \geq 7$. Then we have the following.
        \begin{enumerate}
          \item If $v=2^{(m+1)/2} -1$, then $\gcd(v,n)=1$ and
        		$$\{av \bmod \ n: 0 \leq a \leq 2^{(m-1)/2} +2 \} \subseteq	\Bbb Z_n \setminus T_{(0,m)}.$$
          \item If $v=2^{(m-1)/2} -1$, then $\gcd(v,n)=1$ and $$\{av : 0 \leq a \leq 2^{(m-1)/2} +2 \} \subseteq \Bbb Z_n \setminus T_{(1,m)}.$$
          \item If $v=2^{(m+1)/2} -1$, then $\gcd(v,n)=1$ and $$\{av  \bmod \ n: 2^{m}-2^{(m-1)/2} -5 \leq a \leq n-1, \text{ or } 0 \leq a \leq 2^{(m-1)/2} + 4 \} \subseteq \Bbb Z_n \setminus T_{(2,m)}.$$
        \end{enumerate}		
	\end{lemma}
	
		\begin{lemma}\label{lemma-m63d}
		Let $m \equiv 3 \pmod{6} \geq 9$. Then we have the following.
        \begin{enumerate}
          \item If $v=2^{(m-1)/2} -1$, then $\gcd(v,n)=1$ and $$\{av  \bmod \ n: 2^{m}-2^{(m-1)/2} -5 \leq a \leq n-1, \text{ or } 0 \leq a \leq 2^{(m-1)/2} + 4 \} \subseteq \Bbb Z_n \setminus T_{(0,m)}.$$
          \item If $v=2^{(m+1)/2} -1$, then $\gcd(v,n)=1$ and $$\{av : 0 \leq a \leq 2^{(m-1)/2}  \} \subseteq \Bbb Z_n \setminus T_{(1,m)}.$$
          \item If $v=2^{(m-1)/2} -1$, then $\gcd(v,n)=1$ and $$\{av : 0 \leq a \leq 2^{(m-1)/2}  \} \subseteq \Bbb Z_n \setminus T_{(2,m)}.$$
        \end{enumerate}		
	\end{lemma}
	
		\begin{lemma}\label{lemma-m65d}
		Let $m \equiv 5 \pmod{6} \geq 5$. Then we have the following.
        \begin{enumerate}
          \item If $v=2^{(m-1)/2} -1$, then $\gcd(v,n)=1$ and
        		$$\{av : 0 \leq a \leq 2^{(m-1)/2} +2 \} \subseteq \Bbb Z_n \setminus T_{(0,m)}.$$
          \item If $v=2^{(m-1)/2} -1$, then $\gcd(v,n)=1$ and $$\{av \bmod \ n: 2^{m}-2^{(m-1)/2} -1 \leq a \leq n-1, \text{ or } 0 \leq a \leq 2^{(m-1)/2} \} \subseteq \Bbb Z_n \setminus T_{(1,m)}.$$
          \item If $v=2^{(m+1)/2} -1$, then $\gcd(v,n)=1$ and $$\{av  \bmod \ n: 0 \leq a \leq 2^{(m-1)/2} + 2 \} \subseteq \Bbb Z_n \setminus T_{(2,m)}.$$
        \end{enumerate}		
	\end{lemma}

The dimensions and lower bounds on the minimum distances of the dual codes $\mathcal C_{(i,m)}^\bot$ for $i=0, 1, 2$ are documented in the following theorem.

	\begin{theorem}\label{them-5.4}
		Let $ m \geq 5$ be an odd integer.
\begin{enumerate}
  \item The code $\mathcal C_{(0,m)}^{\bot}$ has parameters $[2^m-1,  k^{\bot}, d^{\bot}]$,
		where $k^\bot$ and $d^\bot$ are given as follows:
 $$k^{\bot}=\begin{cases}
			   	    (2^{m}-2)/3, & \text{if } m \equiv 1 \text{ or } 5 \pmod{6} \geq 5;   \\
		    		(2^{m} -8)/3, & \text{if } m \equiv 3 \pmod{6} \geq 9;  \\
		    		 	\end{cases}$$
	    and $$ d^{\bot} \geq \begin{cases}
	    2^{(m-1)/2} +4, &  \text{if } m \equiv 1 \text{ or } 5 \pmod{6} \geq 5; \\
	   	2^{(m+1)/2} +10, &  \text{if } m \equiv 3 \pmod{6} \geq 9.
	   \end{cases}$$
  \item The code $\mathcal C_{(1,m)}^{\bot}$ has parameters $[2^m-1, k^{\bot}, d^{\bot}]$,
		where $k^\bot$ and $d^\bot$ are given as follows:
$$k^{\bot}=\begin{cases}
			   	    (2^{m}-2)/3, & \text{if } m \equiv 1 \text{ or } 5 \pmod{6} \geq 5;   \\
		    		(2^{m} + 1)/3, & \text{if } m \equiv 3 \pmod{6} \geq 9;  \\
		    		 	\end{cases}$$
	    and
$$ d^{\bot} \geq \begin{cases}
			2^{(m-1)/2} +4, &  \text{if } m \equiv 1  \pmod{6} \geq 7; \\
			2^{(m-1)/2} +2, &  \text{if } m \equiv 3 \pmod{6} \geq 9;\\
			2^{(m+1)/2} +2, &  \text{if } m \equiv 5 \pmod{6} \geq 5.
		\end{cases}$$
  \item The code $\mathcal C_{(2,m)}^{\bot}$ has parameters $[2^m-1,  k^{\bot}, d^{\bot}]$,
	    where $k^\bot$ and $d^\bot$ are given as follows:
$$k^{\bot}=\begin{cases}
			   	    (2^{m}-2)/3, & \text{if } m \equiv 1 \text{ or } 5 \pmod{6} \geq 5;   \\
		    		(2^{m} + 1)/3, & \text{if } m \equiv 3 \pmod{6} \geq 9;  \\
		    		 	\end{cases}$$
	    and
 $$ d^{\bot} \geq \begin{cases}
	    	2^{(m+1)/2} +10, &  \text{if } m \equiv 1 \pmod{6} \geq 7; \\
	    	2^{(m-1)/2} +2, &  \text{if } m \equiv 3 \pmod{6} \geq 9; \\
	    	2^{(m-1)/2} +4, &  \text{if } m \equiv 5  \pmod{6} \geq 5.
	    \end{cases} $$
\end{enumerate}	
	\end{theorem}

	\begin{proof}
We prove the desired conclusion only for the code $\mathcal C_{(0,m)}^{\bot}$, as the conclusion for $\mathcal C_{(1,m)}^{\bot}$ and $\mathcal C_{(2,m)}^{\bot}$ can be similarly proved.

Note that $w_{2}(i) + w_{2}(n-i)= m$ for each $i$ with $ 1 \leq i \leq  n-1$. Then
\begin{equation} \label{equ-T0m}
        T_{(0,m)}^{-1}=\begin{cases}
	    	T_{(1,m)}, &  \text{if } m \equiv 1  \pmod{6} \geq 7;\\
	    	T_{(0,m)}, &  \text{if } m \equiv 3 \pmod{6} \geq 9;\\
            T_{(2,m)}, &  \text{if } m \equiv 5  \pmod{6} \geq 5.
	    \end{cases} \end{equation}
Let $T_{(0,m)}^\bot$ be the defining set of $\mathcal C_{(0,m)}^{\bot}$. Thus we have
$$\dim(\mathcal C_{(0,m)}^{\bot})=n-|T_{(0,m)}^\bot|=n-|\Bbb Z_n \setminus T_{(0,m)}^{-1}|=|T_{(0,m)}^{-1}|.$$
The desired conclusion on the dimension of $\mathcal C_{(0,m)}^{\bot}$ then follows from \eqref{equ-T0m} and Theorems \ref{them-4.1}, \ref{them-4.2}, and \ref{them-4.3}.
The desired lower bound on the minimum distance of $\mathcal C_{(0,m)}^{\bot}$ is deduced from the BCH bound with the help of Lemmas \ref{lemma-m61d}, \ref{lemma-m63d}, and \ref{lemma-m65d}.
	\end{proof}

\begin{example}\label{exam-j91}
    Let $m=5$ and let $\alpha$ be a generator of  $\Bbb F_{2^5}^*$ with $\alpha^{5}+\alpha^{2}+1 = 0$.
    \begin{itemize}
      \item The codes $\mathcal C_{(0,5)}$  and $\mathcal C_{(0,5)}^{\bot}$  have parameters $[31, 21, 5]$ and $[31, 10, 12]$, respectively,  and are optimal according to
the Database \cite{G}.
      \item  The codes $\mathcal C_{(1,5)}$ and $\mathcal C_{(1,5)}^{\bot}$ have parameters $[31, 21, 5]$ and $[31, 10, 10]$, respectively, and the former code is optimal according to
the Database \cite{G}.
      \item The codes $\mathcal C_{(2,5)}$ and $\mathcal C_{(2,5)}^{\bot}$ have parameters $[31, 21, 5]$ and $[31, 10, 12]$, respectively,  and are optimal according to
the Database \cite{G}.
    \end{itemize}
    \end{example}

\subsubsection{Parameters of  $\mathcal C_{(i,m)}^\bot$ when $m \ge 4$ is even}

Assume that $m \geq 4$ is an even integer.
To develop lower bounds on the minimum distances of $\mathcal C_{(i,m)}^\bot$ in this case, the following two lemmas will be employed later. The two lemmas below can be similarly proved by using the same techniques given in Section \ref{sec-auxi}.
Below we only state the lemmas and omit their proofs.

\begin{lemma}\label{lemma-m61ed}
		Let $m \equiv 0 \pmod{4} \geq 4$. Then we have the following.
        \begin{enumerate}
          \item If $v=2^{(m-2)/2} -1$, then $\gcd(v,n)=1$ and
        		$$\begin{cases}
				\{av : 0 \leq a \leq 2^{(m-2)/2}  \} \subseteq \Bbb Z_n \setminus T_{(1,m)}, & \text{if } m\equiv 0 \pmod{6}; \\
				 \{av : 0 \leq a \leq 2^{(m-2)/2 } +4 \} \subseteq \Bbb Z_n \setminus T_{(2,m)}, & \text{if } m\equiv 2 \pmod{6};  \\
			     \{av : 0 \leq a \leq 2^{(m-2)/2} +4\} \subseteq \Bbb Z_n \setminus T_{(0,m)}, & \text{if } m\equiv 4 \pmod{6}.
			    \end{cases}$$
          \item If $v=2^{(m-2)/2} -1$, then $\gcd(v,n)=1$ and 	
                \begin{align*}&\begin{cases}
				\{av \bmod \ n : 2^{m}-2^{(m-2)/2} -3 \leq a \leq n-1, \text{ or } & \\ 0 \leq a \leq 2^{(m-2)/2} +2 \} \subseteq \Bbb Z_n \setminus T_{(0,m)}, & \text{if } m\equiv 0 \pmod{6}; \\
				 \{av \bmod \ n: 2^{m}-2^{(m-2)/2} -3 \leq a \leq n-1, \text{ or } & \\ 0 \leq a \leq 2^{(m-2)/2 } +2 \} \subseteq \Bbb Z_n \setminus T_{(1,m)}, & \text{if } m\equiv 2 \pmod{6};  \\
			     \{av \bmod \ n: 2^{m}-2^{(m-2)/2} -1 \leq a \leq n-1, \text{ or } & \\  0 \leq a \leq 2^{(m-2)/2} \} \subseteq \Bbb Z_n \setminus T_{(2,m)}, & \text{if } m\equiv 4 \pmod{6}.
			    \end{cases}\end{align*}
          \item If $v=2^{(m+2)/2} -1$, then $\gcd(v,n)=1$ and
            $$\begin{cases}
				\{av \bmod \ n : 0 \leq a \leq 2^{(m-2)/2}  \} \subseteq \Bbb Z_n \setminus T_{(2,m)}, & \text{if } m\equiv 0 \pmod{6}; \\
				 \{av \bmod \ n : 0 \leq a \leq 2^{(m-2)/2 } +4 \} \subseteq \Bbb Z_n \setminus T_{(0,m)}, & \text{if } m\equiv 2 \pmod{6};  \\
			     \{av \bmod \ n : 0 \leq a \leq 2^{(m-2)/2} +4\} \subseteq \Bbb Z_n \setminus T_{(1,m)}, & \text{if } m\equiv 4 \pmod{6}.
			    \end{cases}$$
        \end{enumerate}		
	\end{lemma}
	
	\begin{lemma}\label{lemma-m61ed2}
		Let $m \equiv 2 \pmod{4} \geq 6$. Then we have the following.
        \begin{enumerate}
          \item If $v=2^{(m-4)/2} -1$, then $\gcd(v,n)=1$ and
        		$$\begin{cases}
				\{av : 0 \leq a \leq 2^{(m-4)/2}  \} \subseteq \Bbb Z_n \setminus T_{(2,m)}, & \text{if } m\equiv 0 \pmod{6}; \\
				 \{av : 0 \leq a \leq 2^{(m-4)/2 } +2 \} \subseteq \Bbb Z_n \setminus T_{(0,m)}, & \text{if } m\equiv 2 \pmod{6};  \\
			     \{av : 0 \leq a \leq 2^{(m-4)/2} +2\} \subseteq \Bbb Z_n \setminus T_{(1,m)}, & \text{if } m\equiv 4 \pmod{6}.
			    \end{cases}$$
          \item If $v=2^{(m-4)/2} -1$, then $\gcd(v,n)=1$ and 	
               \begin{align*}&\begin{cases}
				\{av \bmod \ n : 2^{m}-2^{(m-4)/2} -5 \leq a \leq n-1, \text{ or } & \\ 0 \leq a \leq 2^{(m-4)/2} +4 \} \subseteq \Bbb Z_n \setminus T_{(0,m)}, & \text{if } m\equiv 0 \pmod{6}; \\
				 \{av \bmod \ n: 2^{m}-2^{(m-4)/2} -1 \leq a \leq n-1, \text{ or } & \\ 0 \leq a \leq 2^{(m-4)/2 } \} \subseteq \Bbb Z_n \setminus T_{(1,m)}, & \text{if } m\equiv 2 \pmod{6};  \\
			     \{av \bmod \ n: 2^{m}-2^{(m-4)/2} -5 \leq a \leq n-1, \text{ or } & \\ 0 \leq a \leq 2^{(m-4)/2} +4 \} \subseteq \Bbb Z_n \setminus T_{(2,m)}, & \text{if } m\equiv 4 \pmod{6}.
			    \end{cases}\end{align*}
          \item If $v=2^{(m+4)/2} -1$, then $\gcd(v,n)=1$ and
            $$\begin{cases}
				\{av \bmod \ n : 0 \leq a \leq 2^{(m-4)/2}  \} \subseteq \Bbb Z_n \setminus T_{(1,m)}, & \text{if } m\equiv 0 \pmod{6}; \\
				 \{av \bmod \ n : 0 \leq a \leq 2^{(m-4)/2 } +2 \} \subseteq \Bbb Z_n \setminus T_{(2,m)}, & \text{if } m\equiv 2 \pmod{6};  \\
			     \{av \bmod \ n : 0 \leq a \leq 2^{(m-4)/2} +2\} \subseteq \Bbb Z_n \setminus T_{(0,m)}, & \text{if } m\equiv 4 \pmod{6}.
			    \end{cases}$$
        \end{enumerate}		
	\end{lemma}

When $m \geq 4$ is even, the parameters of  $\mathcal C_{(i,m)}^\bot$ for $i=0, 1, 2$ are treated in the following theorem.

    \begin{theorem}\label{them-5.7}
    	Let $ m \geq 4$ be even. Then we have the following.
        \begin{enumerate}
      \item The code $\mathcal C_{(0,m)}^{\bot}$ has parameters $[2^m-1, (2^{m}-4)/3, d^{\bot}]$,
    	where
    $$d^{\bot} \geq \begin{cases}
    		2^{m/2} +6, &  \text{ if } m \equiv 0 \pmod{6}; \\
    		2^{(m-2)/2} +6, &  \text{ if } m \equiv 2 \pmod{6}; \\
    		2^{(m-2)/2} +6, &  \text{ if } m \equiv 4 \pmod{6},
    	\end{cases}$$
    	when $m \equiv 0 \pmod{4}$,  or $$  d^{\bot} \geq \begin{cases}
    	2^{(m-2)/2} +10, &  \text{ if } m \equiv 0 \pmod{6}; \\
    	2^{(m-4)/2} +4, &  \text{ if } m \equiv 2 \pmod{6}; \\
    	2^{(m-4)/2} +4, &  \text{ if } m \equiv 4 \pmod{6},
        \end{cases}$$ when $m \equiv 2 \pmod{4}$.
      \item The code $\mathcal C_{(1,m)}^{\bot}$ has parameters $[2^m-1,  k^{\bot}, d^{\bot}]$,
    	where $$ k^{\bot}=\begin{cases}
    		(2^{m} -1)/3, & \text{if } m \equiv 0 \pmod{6};   \\
    		(2^{m}+2)/3, & \text{if } m \equiv 2 \pmod{6};  \\
    		(2^{m} -4)/3, & \text{if } m \equiv 4 \pmod{6},
    	\end{cases}$$ and
    		$$ d^{\bot} \geq \begin{cases}
    		2^{(m-2)/2} +2, &  \text{ if } m \equiv 0 \pmod{6}; \\
    		2^{m/2} +6, &  \text{ if } m \equiv 2 \pmod{6}; \\
    		2^{(m-2)/2} +6, &  \text{ if } m \equiv 4 \pmod{6},
    	\end{cases}$$ when $m \equiv 0 \pmod{4}$, or
     $$d^{\bot} \geq \begin{cases}
    		2^{(m-4)/2} +2, &  \text{ if } m \equiv 0 \pmod{6}; \\
    		2^{(m-2)/2} +2, &  \text{ if } m \equiv 2 \pmod{6}; \\
    		2^{(m-4)/2} +4, &  \text{ if } m \equiv 4 \pmod{6},
    	\end{cases}$$ when $m \equiv 2 \pmod{4}$.
      \item The code $\mathcal C_{(2,m)}^{\bot}$ has parameters $[2^m-1,  k^{\bot}, d^{\bot}]$,
    	where $$ k^{\bot}=\begin{cases}
    		(2^{m} -1)/3, & \text{if } m \equiv 0 \pmod{6};   \\
    		(2^{m}-4)/3, & \text{if } m \equiv 2 \pmod{6};   \\
    		(2^{m} +2)/3, & \text{if } m \equiv 4 \pmod{6},
    	\end{cases}$$ and
    		$$ d^{\bot} \geq \begin{cases}
    		2^{(m-2)/2} +2, &  \text{ if } m \equiv 0 \pmod{6}; \\
    		2^{(m-2)/2} +6, &  \text{ if } m \equiv 2 \pmod{6}; \\
    		2^{m/2} +2, &  \text{ if } m \equiv 4 \pmod{6},
    	\end{cases}$$ when $m \equiv 0 \pmod{4}$, or
    $$  d^{\bot} \geq \begin{cases}
    		2^{(m-4)/2} +2, &  \text{ if } m \equiv 0 \pmod{6}; \\
    		2^{(m-4)/2} +4, &  \text{ if } m \equiv 2 \pmod{6}; \\
    		2^{(m-2)/2} +10, &  \text{ if } m \equiv 4 \pmod{6},
    	\end{cases}$$ when $m \equiv 2 \pmod{4}$.
    \end{enumerate}
       \end{theorem}

    \begin{proof}
    	The proof is very similar to that of Theorem  \ref{them-5.4} and omitted here.
    \end{proof}

	\begin{example}
 Let $m=6$ and  let $\alpha$ be a generator of  $\Bbb F_{2^6}^*$ with $\alpha^{6}+\alpha^{4}+ \alpha^{3}+ \alpha +1 = 0$.
 \begin{itemize}
  \item The codes $\mathcal C_{(0,6)}$  and $\mathcal C_{(0,6)}^{\bot}$  have parameters $[63, 43, 6]$ and $[63, 20, 14]$, respectively.
  \item  The codes $\mathcal C_{(1,6)}$ and $\mathcal C_{(1,6)}^{\bot}$ have parameters $[63, 42, 6]$ and $[63, 21, 16]$, respectively.
  \item The codes $\mathcal C_{(2,6)}$ and $\mathcal C_{(2,6)}^{\bot}$ have parameters $[63, 42, 6]$ and $[63, 21, 16]$, respectively.
 \end{itemize}
\end{example}

\subsection{Comments on the binary cyclic codes $ \mathcal C_{(i,m)} $}

The dimensions $\dim(\mathcal C_{(i,m)}) $ of these codes are around $2n/3$ and the dimensions of their duals are around $n/3$.
Experimental data shows that the lower bounds on the minimum distances of $ \mathcal C_{(i,m)} $ and $ \mathcal C_{(i,m)}^\perp$
developed in the previous subsections are very good. In particular, the three families of codes contain distance-optimal codes
(see Example  \ref{exam-j91}). Hence, they are three dully-good infinite families of binary cyclic codes.

\section{Two families of binary duadic codes with a square-root-like lower bound}\label{sec-duadicc}

\subsection{Known binary duadic codes with a square-root-like lower bound}

In this subsection, we introduce binary duadic codes and survey binary duadic codes with a square-root-like lower bound on
their minimum distances, which are binary duadic codes with parameters of the form $[n, (n \pm1 )/2, d]$ such that $d$ is
very close to $\sqrt{n}$.

		Let $n$ be an odd positive integer and let $\Bbb Z_n$ denote the ring of integers modulo $n$.
	Let $m=\text{ord}_n(2)$, i.e., the order of $2$ modulo $n$. Let $\beta$ be an $n$-th primitive root of unity in
	$\Bbb F_{2^m}$.
	Let $S_1$ and $S_2$
	be two subsets of $\Bbb Z_n$ such that
	\begin{itemize}
		\item $S_1 \cap S_2 = \emptyset$ and $S_1 \cup S_2=\Bbb Z_n \setminus \{0\}$, and
		\item both $S_1$ and $S_2$ are the union of some $2$-cyclotomic cosets modulo $n$.
	\end{itemize}
	If there is a unit $\mu \in \Bbb Z_n$ such that $S_1 \mu = S_2$ and $S_2 \mu =S_1$,
	then $(S_1, S_2, \mu)$ is called a \emph{splitting\index{splitting}} of $\Bbb Z_n$.
	
	Let $(S_1, S_2, \mu)$ be a
	splitting of $\Bbb Z_n$.
	Define
	$$
	g_i(x)=\prod_{i \in S_i} (x - \beta^i) \ \mbox{ and } \ \tilde{g}_i(x)=(x-1) g_i(x)
	$$
	for $i \in \{1,2\}$. The pair of cyclic codes $\mathcal C_1$ and $\mathcal C_2$ of length $n$ over $\Bbb F_2$ with generator
	polynomials $g_1(x)$ and $g_2(x)$ are called \emph{odd-like duadic codes\index{odd-like duadic codes}},
	and the pair of cyclic codes $\tilde{\mathcal C}_1$ and $\tilde{\mathcal C}_2$ of length $n$ over $\Bbb F_2$ with generator
	polynomials $\tilde{g}_1(x)$ and $\tilde{g}_2(x)$ are called \emph{even-like duadic codes\index{even-like duadic codes}}.
	
	By definition, $\mathcal C_1$ and $\mathcal C_2$ have parameters $[n, (n+1)/2]$ and $\tilde{\mathcal C}_1$ and $\tilde{\mathcal C}_2$ have
	parameters $[n, (n-1)/2]$. For odd-like duadic codes, we have the following result \cite[Theorem 6.5.2]{HP03}.
	
	\begin{theorem}[Square root bound]\label{thm-srb}
		Let $\mathcal C_1$ and $\mathcal C_2$ be a pair of odd-like duadic codes of length $n$ over $\Bbb F_2$. Let $d_o$ be their
		(common) minimum odd weight. Then the following hold:
		\begin{enumerate}
			\item $d_o^2 \ge n$.
			\item If the splitting defining the duadic codes is given by $\mu=-1$, then $d_o^2-d_o+1 \geq n$.
			\item Suppose $d_o^2-d_o+1 = n$, where $d_o >2$, and assume that the splitting defining the duadic codes is given by $\mu=-1$. Then $d_o$ is the minimum weight of both $\mathcal C_1$ and $\mathcal C_2$.
		\end{enumerate}
	\end{theorem}

Binary duadic codes are theoretically attractive due to the following facts \cite{DingPless99}:
\begin{itemize}
\item All code examples demonstrate that duadic codes of prime lengths have a square-root bound on their minimum distances.
It is known that duadic codes that are not quadratic-residue codes have a square-root bound on their minimum odd weight
(see Theorem \ref{thm-srb}).
\item Duadic codes could be the best class of cyclic codes of certain lengths. For example, consider $n=31$   and all binary cyclic codes of length $31$    and dimension  $16$. There are four $[31,16]$  binary cyclic codes up to equivalence,
one is the quadratic-residue code and the other is duadic, not quadratic-residue. Both of these have minimum weight $7$, the other two codes have minimum weights $5$ and $6$ \cite[Chapter 6]{HP03}.
\item Another example of very good duadic codes that are not quadratic-residue codes are the  $[113,57,18]$ codes, which have higher minimum weight than the quadratic residue code of the same length \cite{Ples87}. Further the  $[151,76,23]$  and $[233, 117, 32]$
duadic codes have higher minimum weights than the quadratic residue codes of the same lengths \cite{Ples87} and indeed are the best codes known of their length.
\item Every self-dual extended cyclic binary code is the extended code of a duadic code \cite[p. 10]{Ples87}.
\end{itemize}

Duadic codes are a generalisation of the quadratic residue codes. They were introduced
and investigated by Leon, Masley and Pless \cite{Leon84}, Leon \cite{Leon88},
and Pless, Masley and Leon \cite{Ples87}, where a
number of properties are proved. Also all binary duadic codes of length
until 241 are described in \cite{Ples87}.
The total number of binary duadic codes of prime power lengths and their constructions were
presented in  \cite{DLX99} and \cite{DingPless99}.
Further information on the existence, constructions, and properties of duadic codes can
be found in \cite[Chapter 6]{HP03}.

It is proven that binary duadic codes of length
$n$ exist if and only if $n=\prod_i p_i^{m_i}$ where each $p_i \equiv
\pm 1 \pmod{8}$ \cite{Leon84}. In general, there are many binary duadic codes of length
$n$  (see \cite{DLX99} and \cite{DingPless99}).
Experimental computations show that the
minimum distance of many binary duadic codes is poor \cite[Chapter 6]{HP03}.
The minimum weight of an odd-like duadic code may be even. It is open which
binary duadic codes have an odd minimum weight.
Hence, the
lower bound in Theorem \ref{thm-srb} cannot be used to develop a lower bound
on the minimum weight of a binary duadic code.
The only known infinite families of binary duadic codes with a square-root-like lower bound on their minimum distances
are the following:
\begin{itemize}
\item Binary quadratic residue codes with parameters $[n, (n+1)/2, d]$, where $d^2 \geq n$ and $n \equiv \pm 1 \pmod{8}$ is a prime.
\item The punctured binary Reed-Muller codes of order $(m-1)/2$ which has parameters $[2^m-1, 2^{m-1}, 2^{(m+1)/2}-1]$, where $m$
          is odd.
\item  Two infinite families of binary duadic codes with parameters $[2^m-1, 2^{m-1}, d]]$ presented in \cite{TD}, where $m>3$ is odd
and $d$ has a square-root-like lower bound.
\end{itemize}
Motivated by these facts above, we will present two more infinite families of binary duadic codes with a square-root-like lower bound on their minimum distances
in the next subsection.

\subsection{The two families of binary duadic codes}

From now on, we fix $n=2^m-1$ and let $\alpha$ denote a primitive element of $\Bbb F_{2^m}$, where $m \geq 3$ is odd.
Define a polynomial
\begin{equation} \label{equ-gi1i2m}
	g_{(i_1,i_2,m)}(x) =  \prod\limits_{
		\substack{1 \leq i \leq n-1\\
			w_{2}(j) \equiv i_1 \text{ or } i_2 \pmod{4}}}
	(x-\alpha^{j}),
\end{equation}
where $i_1$ and $i_2$ are a pair of distinct elements in the set $\{0,1,2,3\}$.
It is easily seen that $g_{(i_1,i_2,m)}(x) \in \Bbb F_2[x]$. Let $\mathcal C_{(i_1,i_2,m)}$ denote the binary cyclic code of length $n=2^m-1$ with generator polynomial $g_{(i_1,i_2,m)}(x)$.
 Denote
   $$T_{(i_1,i_2,m)} = \{ 1 \leq j \leq n-1: \ \ w_{2}(j) \equiv i_1 \text{ or } i_2 \pmod{4}\}$$
   for $i=0, 1, 2,3$.  It is  clear that $T_{(i_1,i_2,m)}$ is the defining set of $\mathcal C_{(i_1,i_2,m)}$ with respect to the $n$-th primitive root of unity $\alpha$. When $m$ is odd, it is easy to see that $\mathcal C_{(0,2,m)}$ and $\mathcal C_{(1,3,m)}$ form a  pair of odd-like duadic codes, which were originally proposed and studied in \cite{TD}. We will prove the following two
   statements later:
\begin{itemize}
   	\item  When $m \equiv 1 \pmod 4$, $\mathcal C_{(0,3,m)}$ and $\mathcal C_{(1,2,m)}$ form a  pair of odd-like duadic codes.
   	\item  When $m \equiv 3 \pmod 4$, $\mathcal C_{(0,1,m)}$ and $\mathcal C_{(2,3,m)}$ form a  pair of odd-like duadic codes.
   \end{itemize}
In the following subsections, we will investigate the dimensions and minimum distances of the two families of duadic codes and their dual and extended codes.

\subsection{Some auxiliary results}

	To develop lower bounds on the minimum distances of the two families of duadic codes, we need some auxiliary results about their defining sets.
	
	\begin{lemma}\label{lemma-m81}
		Let $m \equiv 1 \pmod{8} \geq 9$. Then we have the following.
\begin{enumerate}
  \item If $v=2^{(m-1)/2} -1$, then $\gcd(v,n)=1$ and
		$$\{av : 1 \leq a \leq 2^{(m-1)/2} +2 \} \subseteq	 T_{(0, 3, m)}.$$
  \item If $v=2^{(m+1)/2} -1$, then $\gcd(v,n)=1$ and $$\{av : 1 \leq a \leq 2^{(m-1)/2} +2 \} \subseteq  T_{(1, 2, m)}.$$
\end{enumerate}		
	\end{lemma}
	
	\begin{proof}
If $v=2^{(m-1)/2} -1$, it follows from Lemma \ref{lemma-gcd2} that $ \gcd(v,n) =1$.
		 When $a=2^{(m-1)/2} +2$, we have
		$$av =2^{m-1} + 2^{(m-1)/2} -2 = 2(2^{m-2} + 2^{(m-3)/2}-1).$$
		Consequently, $w_{2}(av) =(m-1)/2 \equiv 0 \pmod{4}$.
		When $ a=2^{(m-1)/2} +1 $, $ av = 2^{m-1} -1$ and $w_{2}(av) =m-1 \equiv 0 \pmod{4}$.
		When $ a=2^{(m-1)/2}$, $w_{2}(av) =w_{2}(v)=(m-1)/2\equiv 0 \pmod{4}$.
		Now we assume that $1 \leq a \leq 2^{(m-1)/2} - 1$. Let $a=2^{l}\bar{a}$, where $\bar{a}$ is odd and $l \ge 0$ is an integer. Then we have $1 \leq \bar{a} \leq 2^{(m-1)/2} - 1$ and the $2$-adic expansion of $\bar{a}$ is given by $$ \bar{a} = \sum\limits_{i=0}^{(m-3)/2}a_{i}2^{i}.$$
		Since $\bar{a}$ is odd, $a_{0}=1$. We have
		\begin{align*}
			\bar{a}v =\sum\limits_{i=1}^{(m-3)/2}a_{i}2^{i+(m-1)/2} + \sum\limits_{i=0}^{(m-3)/2}(1-a_{i})2^{i} +1.
		\end{align*}
	    It then follows that
		$$w_{2}(\bar{a}v) =  w_{2}(\bar{a})-1 + 1+ \frac{m-1}{2}-w_{2}(\bar{a}) = \frac{m-1}{2}\equiv 0 \pmod{4}.$$
		The desired conclusion in the first case then follows.

		If $v=2^{(m+1)/2} -1$, it follows from Lemma \ref{lemma-gcd2} that
		$$\gcd(v,n)=2^{\gcd((m+1)/2,m)}-1 =2^{\gcd((m+1)/2,(m-1)/2)}-1 =1.$$
	When $ a=2^{(m-1)/2}$, it is easy to see that
		$$w_{2}(av) = w_{2}(v) = \frac{m+1}{2}\equiv 1 \pmod{4}.$$
Furthermore, one can similarly check that $w_{2}(av) \equiv 1 \pmod{4}$ for $a=2^{(m-1)/2}+1$ and $2^{(m-1)/2}+2$.
		Next, we assume that $1 \leq a \leq 2^{(m-1)/2} - 1$. Let $a=2^{l}\bar{a}$, where $\bar{a}$ is odd and $l \ge 0$ is an integer. Then we have $1 \leq \bar{a} \leq 2^{(m-1)/2} - 1$. Let the 2-adic expansion of $\bar{a}$ be given by
		$$\bar{a} = \sum\limits_{i=0}^{(m-3)/2}a_{i}2^{i}.$$
		Since $\bar{a}$ is odd, $a_{0}=1$.  Then
		\begin{align*}
			\bar{a}v &= \bar{a}2^{(m+1)/2} -\bar{a} \\
			&= \sum\limits_{i=1}^{(m-3)/2}a_{i}2^{i+(m+1)/2} +2^{(m-1)/2} + \sum\limits_{i=0}^{(m-3)/2}(1-a_{i})2^{i} +1.
		\end{align*}
		As a result, we have
		$$ w_{2}(\bar{a}v) =  w_{2}(\bar{a})-1 + 2 + \frac{m-1}{2}-w_{2}(\bar{a}) = \frac{m+1}{2}\equiv 1 \pmod{4}.$$
		This completes the proof.
	\end{proof}	
	
The following three lemmas  can be similarly proved and their proofs are omitted here.
		
		\begin{lemma}\label{lemma-m83}
		Let $m \equiv 3 \pmod{8} \geq 3$. Then we have the following.
 \begin{enumerate}
  \item If $v=2^{(m-1)/2} -1$, then $\gcd(v,n)=1$ and
		$$\{av : 1 \leq a \leq 2^{(m-1)/2}  \} \subseteq  T_{(0, 1, m)}.$$
  \item If $v=2^{(m+1)/2} -1$, then $\gcd(v,n)=1$ and $$\{av : 1 \leq a \leq 2^{(m-1)/2} \} \subseteq  T_{(2, 3, m)}.$$
\end{enumerate}		
	\end{lemma}

		\begin{lemma}\label{lemma-m85}
		Let $m \equiv 5 \pmod{8} \geq 5$. Then we have the following.
 \begin{enumerate}
  \item If $v=2^{(m+1)/2} -1$, then $\gcd(v,n)=1$ and
		$$\{av : 1 \leq a \leq 2^{(m-1)/2}  \} \subseteq  T_{(0, 3, m)}.$$
  \item If $v=2^{(m-1)/2} -1$, then $\gcd(v,n)=1$ and $$\{av : 1 \leq a \leq 2^{(m-1)/2} \} \subseteq  T_{(1, 2, m)}.$$
\end{enumerate}		
	\end{lemma}

		\begin{lemma}\label{lemma-m87}
		Let $m \equiv 7 \pmod{8} \geq 7$. Then we have the following.
 \begin{enumerate}
  \item If $v=2^{(m+1)/2} -1$, then $\gcd(v,n)=1$ and
		$$\{av : 1 \leq a \leq 2^{(m-1)/2}+2 \} \subseteq  T_{(0, 1, m)}.$$
  \item If $v=2^{(m-1)/2} -1$, then $\gcd(v,n)=1$ and $$\{av : 1 \leq a \leq 2^{(m-1)/2}+2 \} \subseteq  T_{(2, 3, m)}.$$
\end{enumerate}		
	\end{lemma}

\subsection{Parameters of  the two families of duadic codes and their related codes}

When $m \equiv 1 \pmod 4\geq 5$, the parameters of  $\C_{(0,3,m)}$ and $\C_{(1,2,m)}$  are treated in the following theorem.

\begin{theorem}\label{thm-m41}
	Let $m \equiv 1 \pmod 4\geq 5$ be an integer. Then $\C_{(0,3,m)}$ and $\C_{(1,2,m)}$ form a pair of odd-like duadic codes with
	parameters $[2^m-1, 2^{m-1}, d]$, where
	\begin{eqnarray*}
		d \geq \left\{
		\begin{array}{ll}
			2^{(m-1)/2}+3 & \mbox{ if } m \equiv 1 \pmod{8}, \\
			2^{(m-1)/2}+1 & \mbox{ if } m \equiv 5 \pmod{8}.
		\end{array}
		\right.
	\end{eqnarray*}
\end{theorem}

\begin{proof}
	Note that $m \equiv 1 \pmod 4$ and $w_2(i)=m-w_2(n-i)$ for each $i$ with $1 \leq i \leq n-1$.
	Hence, $i \in T_{(0,3,m)}$ if and only if  $n-i \in T_{(1,2,m)}$. It then follows that $T_{(0,3,m)}$ and
	$T_{(1,2,m)}$ partition $\Bbb Z_n \setminus \{0\}$ and
	$$
	T_{(0,3,m)}=-T_{(1,2,m)} \mbox{ and } T_{(1,2,m)}=-T_{(0,3,m)}.
	$$
	It then follows that $\C_{(0,3,m)}$ and $\C_{(1,2,m)}$ form a pair of duadic codes with length $n$ and
	dimension $(n+1)/2$. Hence, the two codes have the same minimum distance $d$.
	
	We only prove the lower bounds on minimum distannce $d$ in the case that $m \equiv 1 \pmod 8$ as it is similar to prove the  desired conclusion for $m \equiv 5 \pmod 8$. Denote $v=2^{\frac {m-1} 2}-1$. It follows from Lemma \ref{lemma-m81} that $\gcd(v,n) =1$. Let $ \overline{v}$ be the integer satisfying $v \overline{v} \equiv 1 \pmod n$.
	Write $\gamma = \alpha^{\overline{v}}$. It is deduced from Lemma \ref{lemma-m81} that defining set of $\mathcal C_{(0,3,m)}$ with respect to $\gamma$ contains the set $\left\{ 1,2,...,2^{(m-1)/2}+2 \right\}$. The desired lower bound on the minimum distance of $\mathcal C_{(0,3,m)}$ then follows from the BCH bound on cyclic codes. The desired conclusion on the minimum distance of $\mathcal C_{(1,2,m)}$  follows naturally, as the two duadic codes have the same minimum distance. This completes the proof.
\end{proof}

When $m \equiv 1 \pmod 4\geq 5$, the parameters of  the dual codes $\C_{(0,3,m)}^\perp$ and $\C_{(1,2,m)}^\perp$  are studied in the following theorem.
\begin{theorem}\label{thm-m41-dual}
	Let $m \equiv 1 \pmod 4\geq 5$ be an integer. Then $\C_{(0,3,m)}^\perp$ and $\C_{(1,2,m)}^\perp$ form a pair of even-like duadic codes with
	parameters $[2^m-1, 2^{m-1}-1, d^\perp]$, where
	\begin{eqnarray*}
		d^\perp \geq \left\{
		\begin{array}{ll}
			2^{(m-1)/2}+4 & \mbox{ if } m \equiv 1 \pmod{8}, \\
			2^{(m-1)/2}+2 & \mbox{ if } m \equiv 5 \pmod{8}.
		\end{array}
		\right.
	\end{eqnarray*}
\end{theorem}

\begin{proof}
	Note that the defining sets of $\C_{(0,3,m)}^\perp$ and $\C_{(1,2,m)}^\perp$ with respect to $\alpha$ are
	$\{0\} \cup T_{(0,3,m)}$ and $\{0\} \cup T_{(1,2,m)}$, respectively. It then follows that  $\C_{(0,3,m)}^\perp$ is
	the even-weight subcode of $\C_{(0,3,m)}$ and $\C_{(1,2,m)}^\perp$ is the even-weight subcode of $\C_{(1,2,m)}$.
	The desired conclusion then follows from Theorem \ref{thm-m41}.
\end{proof}

When $m \equiv 1 \pmod 4\geq 5$, the parameters of extended codes $\overline{\C_{(0,3,m)}}$ and $\overline{\C_{(1,2,m)}}$  are investigated in the following theorem.

\begin{theorem}\label{thm-m41-extended}
	Let $m \equiv 1 \pmod 4\geq 5$ be an integer. Then the extended codes $\overline{\C_{(0,3,m)}}$ and $\overline{\C_{(1,2,m)}}$ of
	$\C_{(0,3,m)}$ and $\C_{(1,2,m)}$ are self-dual and doubly-even. Furthermore, they have parameters
	$[2^m, 2^{m-1}, \bar{d} \ge 2^{(m-1)/2}+4]$, where $\bar{d}$ denotes the minimum distance of $\overline{\C_{(0,3,m)}}$ or $\overline{\C_{(1,2,m)}}$.
\end{theorem}

\begin{proof}
	It is well known that the extended codes of a pair of odd-like binary duadic codes are self-dual if
	the splitting corresponding to the pair of odd-like binary duadic codes is given by $-1$ \cite[Theorem 6.4.12]{HP03}. Note that
	$(T_{(0,3,m)}, T_{(1,2,m)}, -1)$ is a splitting of $\Bbb Z_{2^m-1}$ when $m \equiv 1 \pmod 4$. Consequently, the extended codes
	$\overline{\C_{(0,3,m)}}$ and $\overline{\C_{(1,2,m)}}$ of $\C_{(0,3,m)}$ and $\C_{(1,2,m)}$ are self-dual.
	It then follows from \cite[Theorem 6.5.1]{HP03} that the Hamming weight of each codeword
	in $\overline{\C_{(0,3,m)}}$ and $\overline{\C_{(1,2,m)}}$ is divisible by $4$. The remaining conclusions follow from Theorem \ref{thm-m41}.
\end{proof}

Similarly, we have the following three theorems on parameters of $\C_{(0,1,m)}$ and $\C_{(2,3,m)}$ and their dual and extended codes when $m \equiv 3 \pmod 4$.

\begin{theorem}\label{thm-m43}
	Let $m \equiv 3 \pmod 4$ be an integer. Then $\C_{(0,1,m)}$ and $\C_{(2,3,m)}$ form a pair of odd-like duadic codes with
	parameters $[2^m-1, 2^{m-1}, d]$, where
	\begin{eqnarray*}
		d \geq \left\{
		\begin{array}{ll}
			2^{(m-1)/2}+1 & \mbox{ if } m \equiv 3 \pmod{8}, \\
			2^{(m-1)/2}+3 & \mbox{ if } m \equiv 7 \pmod{8}.
		\end{array}
		\right.
	\end{eqnarray*}
\end{theorem}

\begin{proof}
	The proof is very similar to that of Theorem \ref{thm-m41} and omitted here.
\end{proof}

\begin{theorem}\label{thm-m43-dual}
	Let $m \equiv 3 \pmod 4$ be an integer. Then $\C_{(0,1,m)}^\perp$ and $\C_{(2,3,m)}^\perp$ form a pair of even-like duadic codes with
	parameters $[2^m-1, 2^{m-1}-1, d^\perp]$, where
	\begin{eqnarray*}
		d^\perp \geq \left\{
		\begin{array}{ll}
			2^{(m-1)/2}+2 & \mbox{ if } m \equiv 3 \pmod{8}, \\
			2^{(m-1)/2}+4 & \mbox{ if } m \equiv 7 \pmod{8}.
		\end{array}
		\right.
	\end{eqnarray*}
\end{theorem}

\begin{proof}
	The proof is very similar to that of Theorem \ref{thm-m41-dual} and omitted here.
\end{proof}

\begin{theorem}\label{thm-m43-extended}
	Let $m \equiv 3 \pmod 4 \ge 7$ be an integer. Then the extended codes $\overline{\C_{(0,1,m)}}$ and $\overline{\C_{(2,3,m)}}$ of
	$\C_{(0,1,m)}$ and $\C_{(2,3,m)}$ are self-dual and doubly-even. Furthermore, they have parameters
	$[2^m, 2^{m-1},  \bar{d} \ge 2^{(m-1)/2}+4]$, where $\bar{d}$ denotes the minimum distance of $\overline{\C_{(0,1,m)}}$ or $\overline{\C_{(2,3,m)}}$.
\end{theorem}

\begin{proof}
	The proof is very similar to that of Theorem \ref{thm-m41-extended} and omitted here.
\end{proof}

\begin{example}\label{exam-J101}
	Let $m=5$ and  let $\alpha$ be a generator of  $\Bbb F_{2^5}^*$ with $\alpha^{5}+\alpha^{2}+1 = 0$.
	\begin{itemize}
		\item The codes $\mathcal \C_{(0,3,5)}$  and $\C_{(0,3,5)}^\perp $  have parameters $[31, 16, 7]$ and $[31, 15, 8]$, respectively,  where the former code is almost optimal in the sense that the minimum distance of the optimal binary
linear code with length $31$ and dimension $16$ is $8$, and the latter code is optimal according to the Database \cite{G}.
In fact, $\mathcal \C_{(0,3,5)}$ has the same parameters as the best binary cyclic code according to Table A.13 in
\cite{DingBK18}.
		\item  The codes $\mathcal \C_{(1,2,5)}$  and $\C_{(1,2,5)}^\perp $ have parameters $[31, 16, 7]$ and $[31, 15, 8]$, respectively. The comments on the parameters of $\mathcal \C_{(0,3,5)}$  and $\C_{(0,3,5)}^\perp $ above apply here to
		the parameters of $\mathcal \C_{(1,2,5)}$  and $\C_{(1,2,5)}^\perp $.
	\end{itemize}
\end{example}

\begin{example}\label{exam-J102}
	Let $m=7$ and  let $\alpha$ be a generator of  $\Bbb F_{2^7}^*$ with $\alpha^{7}+\alpha+1 = 0$.
	\begin{itemize}
		\item The codes $\mathcal \C_{(0,1,7)}$  and $\C_{(0,1,7)}^\perp $  have parameters $[127, 64, 15]$ and $[127, 63, 20]$, respectively.
		\item  The codes $\mathcal \C_{(2,3,7)}$ and $\C_{(2,3,7)}^\perp $ have parameters $[127, 64, 15]$ and $[127, 63, 20]$, respectively.
	\end{itemize}
	Note that $\mathcal \C_{(0,1,7)}$ and $\mathcal \C_{(2,3,7)}$ have the same parameters as the punctured binary Reed-Muller code
$\mathrm{PRM}_2(3, 7)$ of length $127$ and order $3$.  	
	The best binary duadic code known of length $127$ and dimension $64$ has minimum distance $19$ \cite{TD}.
\end{example}

The lower bounds on the minimum distances of $\C_{(0,1,m)}$, $\C_{(2,3,m)}$, $\C_{(0,3,m)}$, $\C_{(1,2,m)}$ and their
duals developed in this paper are very close to the square-root bound. Hence, they are very good codes in general. Example \ref{exam-J101}
shows that  $\C_{(0,3,m)}$ and $\C_{(1,2,m)}$ could be a best cyclic code.
Example   \ref{exam-J102}
shows that the minimum distance of $\C_{(0,1,7)}$ and $\C_{(2,3,7)}$ is less than that of the two duadic codes
$\C_{(0,2,7)}$ and $\C_{(1,3,7)}$ presented in \cite{TD}, but their duals $\C_{(0,1,7)}^\perp$ and $\C_{(2,3,7)}^\perp$,
$\C_{(0,2,7)}^\perp$ and $\C_{(1,3,7)}^\perp$ have the same parameters $[127, 63, 20]$.
Example   \ref{exam-J102}
shows that the minimum distance of $\C_{(0,1,7)}$ and $\C_{(2,3,7)}$ equals that of
the punctured binary Reed-Muller code
$\mathrm{PRM}_2(3, 7)$ of length $127$ and order $3$, but $\C_{(0,1,7)}^\perp$ and $\C_{(2,3,7)}^\perp$ are much better than
$\mathrm{PRM}_2(3, 7)^\perp$ as
$$
d(\mathrm{PRM}_2(3, 7)^\perp)=16
$$
and
$$
d(\C_{(0,1,7)}^\perp) = d(\C_{(2,3,7)}^\perp) =20,
$$
where $d(\C)$ denotes the minimum distance of $\C$.

\subsection{Differences among several families of duadic codes}

None of the families of duadic codes, $\C_{(0,1,m)}$, $\C_{(2,3,m)}$, $\C_{(0,3,m)}$ and $\C_{(1,2,m)}$, is identical with
the family of binary quadratic residue codes, as $2^m-1$ is composite for many odd $m$. The two families of binary
duadic codes studied in \cite{TD} are $\C_{(0,2,m)}$, $\C_{(1,3,m)}$ and thus are not identical with any of the families of
duadic codes investigated in this paper, as their defining sets are different. This is also justified by the facts that
$$
d(\mathcal \C_{(0,1,7)})=d(\mathcal \C_{(2,3,7)})=15
$$
and
$$
d(\mathcal \C_{(0,2,7)})=d(\mathcal \C_{(1,3,7)})=19,
$$
where $d(\C)$ denotes the minimum distance of $\C$.

The family of punctured binary Reed-Muller codes $\mathrm{PRM}_2((m-1)/2, m)$ of length $2^m-1$ and order $(m-1)/2$
is not identical with any of the families of
duadic codes investigated in this paper, as their defining sets are different. This is also justified by the facts that
$$
d(\mathcal \C_{(0,1,7)}^\perp)=d(\mathcal \C_{(2,3,7)}^\perp)=20
$$
and
$$
d(\mathrm{PRM}_2(3, 7)^\perp)=16.
$$

\section{Summary and concluding remarks}\label{sec-Con}
	 In this paper, we constructed and studied the parameters of the five families of binary cyclic codes with parameters
	 $[n, k,d]$ and their duals, where $n=2^m-1$ and $(n-6)/3 \leq  k \leq (n+6)/3$. They contain some distance-optimal
	 codes and are very good in general, as they have a very good lower bound on their minimum distances. The work on the
	 two families of duadic codes $\C_{(0,1,m)}$, $\C_{(2,3,m)}$, $\C_{(0,3,m)}$ and $\C_{(1,2,m)}$ complements the work
	 in \cite{TD}. It is possible to improve the lower bounds on the binary cyclic codes developed in this paper.
	
	 The works of
	 \cite{TD} and this paper can be generalised and extended to obtain more families of binary duadic codes and other
	 binary cyclic codes. But it will be more difficult to develop a good lower bound on their minimum distances.
	 The generalisation goes as follows.
	 Let $r \geq 2$ be a positive integer and let $n=2^m-1$ for an integer $m \geq 3$.
	 Let $S$ be any proper subset of $\mathbb{Z}_r$.
     Define
     $$
     T_{[r,m,S]}=\{1 \leq i \leq n-1: w_2(i) \bmod{r} \in S\}.
     $$	
     By definition, $T_{[r,m,S]}$ is the union of some $2$-cyclotomic cosets modulo $n$. Let $\alpha$ be a primitive element of
     $\mathbb{F}_{2^m}$.
	 Let $\C_{[r,m,S]}$ denote the binary cyclic code of length $n$ with generator polynomial
	 $$
	 g_{[r,m,S]}(x)=\prod_{i \in T_{[r,m,S]}} (x-\alpha^i).
	 $$
	 When $r=2$ and $|S|=1$, the codes $\C_{[r,m,S]}$ were studied in \cite{TD}. When $r=3$ and  $|S|=1$, the codes $\C_{[r,m,S]}$
	 were treated in this paper.  When $r=4$ and  $|S|=2$, some of the codes $\C_{[r,m,S]}$ were investigated in this paper, and
	 the others were not studied in this paper as they are not duadic codes. When $r \geq 6$ is even and $|S|=r/2$, the code
	 $\C_{[r,m,S]}$ could be a duadic code for certain odd $m$. Hence, we propose the following research problems.
	
	 \begin{problem}
	 Let $m \geq 3$ be an integer and let $r=4$. Let $S \in \{(0,1), (2,3) \}$,
	  Determine the parameters of the code $\C_{[r,m,S]}$ for
	 $m \not\equiv 3 \bmod{4}$.
	 \end{problem}

	 \begin{problem}
	 Let $m \geq 3$ be an integer and let $r=4$. Let $S \in \{(0,3), (1,2) \}$.
	  Determine the parameters of the code $\C_{[r,m,S]}$ for
	 $m \not\equiv 1 \bmod{4}$.
	 \end{problem}  	
	
	 The two problems above can be solved with similar techniques in this paper. The following two research problems are
	 harder.
	
	 \begin{problem}
	 Let $m \geq 3$ be an integer and let $r \geq 5$ be an integer.  For any proper subset $S$ of $\mathbb{Z}_r$,
	 determine the parameters of the code $\C_{[r,m,S]}$.
	 \end{problem}
	
	 	 \begin{problem}
	 Let $r \geq 6$ be an even integer. Find a subset $S$ of $\mathbb{Z}_r$ with $|S|=r/2$
	 such that $\C_{[r,m,S]}$ is a binary duadic code of length $n=2^m-1$ for infinitely many odd $m$.
	 Determine the parameters of these duadic codes.
	 \end{problem}


\begin{thebibliography}{99}





\bibitem{Charpin} P. Charpin, ``Open problems on cyclic codes,'' in Handbook of Coding Theory, vol. 1, V. S. Pless and W. C. Huffman, Eds. Amsterdam,  The Netherlands: Elsevier, 1998, pp. 963--1063.
		



\bibitem{DingBK18}
C. Ding, \emph{Codes from Difference Sets}. Singapore: World Scientific, 2018.


\bibitem{DLX99}
 C. Ding, K. Y. Lam and C. Xing, ``Enumeration and construction of all
      duadic codes of length $p^m$," {\it Fundamenta Informaticae}, vol. 38, no.
      1, pp. 149--161, 1999.


\bibitem{DingPless99}
C. Ding, V. Pless, ``Cyclotomy and duadic codes of prime lengths,"
\emph{IEEE Trans. Inf. Theory}, vol. 45, no. 2, pp. 453--466, March 1999.


\bibitem{DY} C. Ding and J. Yang, ``Hamming weights in irreducible cyclic codes,''
\emph{Discrete Math.}, vol. 313, no. 4, pp. 434--446, 2013.


\bibitem{HP03} W. C. Huffman and V. Pless, \emph{Fundamentals Error-Correcting Codes}.
Cambridge, U.K.: Cambridge Univ. Press, 2003.

\bibitem{GDL} B. Gong, C. Ding and C. Li, ``The dual codes of several classes of BCH codes,''
	    \emph{IEEE Trans. Inf. Theory}, vol. 68, no. 2, pp. 953--964,  2022.

\bibitem{G} M. Grassl, Bounds on the minimum distance of linear codes and
		quantum codes, Online available at http://www.codetables.de, accessed on 2023-1-1.


\bibitem{Leon84} J. S. Leon, J. M. Masley, and V. Pless, ``Duadic codes,"
         \emph{IEEE Trans. Inf. Theory,} vol. 30, no. 5, pp. 709--714, Sept. 1984.

\bibitem{Leon88} J. S. Leon, ``A probabilistic algorithm for computing
         minimum weight of large error-correcting codes," \emph{IEEE Trans
         Inf. Theory,} vol. 34, no. 5,  pp. 1354--1359, Sept. 1988.



\bibitem{LSX} S. Li, ``The minimum distance of some narrow-sense primitive BCH codes,'' \emph{ SIAM J. Discrete Math.},  vol. 31, no. 4, pp. 2530--2569, 2017.

\bibitem{JPen}
S. Noguchi, X.-N. Lu, M. Jimbo, Y. Miao,
``BCH codes with minimum distance proportional to code length,
\emph{SIAM J. Discrete Math.}, vol. 35,  no. 1, pp. 179--193,  2021.

\bibitem{Ples87} V. Pless, J. M. Masley, and J. S. Leon, ``On weights
         in duadic codes," \emph{J. Comb. Theory Ser. A}, vol. 44, pp. 6--21, 1987.


		
\bibitem{TD} C. Tang, C. Ding, ``Binary $[n,(n + 1)/2]$ cyclic codes with good
minimum distances,'' \emph{IEEE Trans. Inf. Theory}, vol. 68, no. 12, pp. 7842--7849,  2022.


\bibitem{SYW} X. Shi, Q. Yue, Y. Wu, ``The dual-containing primitive BCH codes with the maximum designed distance and their applications to quantum codes,''
\emph{Des. Codes and Cryptogr.}, vol. 87, pp. 2165--2183, 2019.

		
\bibitem{Xiong} M. Xiong, ``On cyclic codes of composite length and the minimum distance,"
\emph{IEEE Trans. Inf. Theory}, vol. 64, no. 9, pp. 6305--6314,  2018.
		
\bibitem{XiongZhang}
M. Xiong, A. Zhang, ``On cyclic codes of composite length and the minimum distance II,"
\emph{EEE Trans. Inf. Theory}, vol. 67, no. 8, pp. 5097--5103,  2021.
		
		
	
	\end{thebibliography}
\end{document}